%% file: main.tex
\documentclass{article}

\usepackage{epsfig,amsmath,amstext,psfrag,mathrsfs,amssymb,amsfonts,color,amsthm}
\usepackage{latexsym,algpseudocode,algorithm,algorithmicx,bm}
\usepackage{enumerate,cite}
\usepackage{graphicx,subfigure}
\usepackage{dsfont}
\usepackage{indentfirst}

\DeclareMathOperator*{\argmax}{arg\,max}
\DeclareMathOperator*{\argmin}{arg\,min}

\newtheorem{definition}{\textbf{Definition}}

\newtheorem{lemma}{\textbf{Lemma}}
\newtheorem{theorem}{\textbf{Theorem}}
\newtheorem{proposition}{\textbf{Proposition}}
\newtheorem{remark}{\textbf{Remark}}

\newtheorem{example}{\textbf{Example}}

\newcommand{\nn}{\nonumber}

\newcommand\redsout{\bgroup\markoverwith{\textcolor{red}{\rule[0.5ex]{2pt}{0.8pt}}}\ULon}
\newcommand{\zsfd}[1]{\ifmmode\text{\redsout{\ensuremath{#1}}}\else\redsout{#1}\fi}

\begin{document}

\title{Linear-Complexity Exponentially-Consistent Tests for Universal Outlying Sequence Detection}
\author{Yuheng Bu\qquad
Shaofeng Zou\qquad
Venugopal V. Veeravalli \\
 University of Illinois at Urbana-Champaign\\
Email: bu3@illinois.edu, szou3@illinois.edu,  vvv@illinois.edu
}



%
%


\maketitle

\begin{abstract}
The problem of universal outlying sequence detection is studied, where the goal is to detect outlying sequences among $M$ sequences of samples.
A sequence is considered as outlying if the observations therein are generated by a distribution different from those generating the observations in the majority of the sequences.
In the universal setting, we are interested in identifying all the outlying sequences without knowing the underlying generating distributions.
In this paper,  a class of tests based on distribution clustering is proposed. These tests are shown to be exponentially consistent with \textit{linear} time complexity in $M$.
Numerical results demonstrate that our clustering-based tests achieve similar performance to existing tests, while being considerably more computationally efficient.

\end{abstract}

\section{Introduction}

\let\thefootnote\relax\footnotetext{This work was presented in part at the International Symposium on Information Theory (ISIT), Aachen, Germany, 2017 \cite{bu2017isit}.}
\footnotetext{The work of Y. Bu, S. Zou and V. V. Veeravalli was supported by the National Science Foundation under grants NSF 11-11342 and 1617789, through the University of Illinois at Urbana-Champaign.}

We study a universal outlying sequence detection problem, where the objective is to detect outlying sequences among $M$ sequences of samples. Each sequence consists of $n$ independent and identically distributed (i.i.d.) discrete observations.
It is assumed that the observations in the majority of the sequences are distributed according to typical distributions. A sequence is considered as outlying if its distribution is different from the typical distributions.
We are interested in the universal setting of the problem, where we do not know the probability mass functions (pmfs) of both the typical and outlying distributions, which are assumed to have full support over a finite alphabet. The goal is to design a universal test, which does not depend on the typical and outlying distributions, to discern all the outlying sequences efficiently.


Outlying sequence detection finds possible applications in many domains \cite{vvv2014}. For example, in cognitive wireless networks, channel measurements follow different distributions depending on whether the channel is busy or vacant. In order to utilize the vacant channels for improving spectral efficiency, such a network need to efficiently identify vacant channels out of a large number busy channels based on their corresponding signals. Other applications include anomaly detection in large data sets \cite{bolton2002statistical}, security monitoring in sensor networks \cite{chamberland2007wireless}, and detecting an epidemic disease with aberrant genetic markers \cite{vucetic2003detection}. All of these applications require a reliable algorithm that can be implemented with low time complexity.



In universal outlying sequence detection problem, we have no prior knowledge and no training data to learn these distributions before hand. Thus, the major challenges to solve this problem lie in: (1) building distribution-free consistent tests, and further guaranteeing their exponential consistency for distinct typical and outlying distributions; and (2) designing low-complexity tests that can be used in practical applications.

To overcome theses challenges, we propose tests based on distribution clustering \cite{banerjee2005clustering} for various scenarios.
We show that our tests are exponentially consistent,  with time complexity that is linear in the number of sequences $M$ and independent of the number of outlying sequences $T$.

The intuition for the clustering-based test is that if we observe a sequence of samples from each distribution, the empirical distributions of the sequences will converge to the true distributions as the number of samples goes to infinity. Moreover, we assume that the typical distributions (and also possibly the outlying distributions) form a cluster, which means they are usually closer to each other than to the outlying distributions. This suggests that the outlying sequence detection problem can be solved by clustering the empirical distributions using KL divergence as the distance metric (see also \cite{li2016outlying}).



We note that our clustering-based tests are closely related to the classical distribution clustering problem \cite{banerjee2005clustering,chaudhuri2008finding,ackermann2010clustering,nock2008mixed,xiong2011group}, but there are essential differences. In the distribution clustering problem, the goal is to construct low-complexity algorithm to find the cluster structure of distributions with the lowest cost (sum of distance functions of each distribution in the cluster to the center). Whereas in our problem, we are given samples from each distribution rather than the actual underlying distribution itself. Since we are considering a detection problem, we are more interested in the error probability of our test, rather than the minimal clustering cost. Previous studies on approximation algorithms for distribution clustering \cite{chaudhuri2008finding,ackermann2010clustering,nock2008mixed} only show that by carefully seeding the initialization step, the cost corresponding to the cluster structure returned by the approximation can be bounded within a $\log K$ factor of the minimal cost, where $K$ is the number of clusters. And there are no results showing that the approximation algorithms will converge to the minimal cost. Therefore, their results cannot be directly applied to our problem to provide statistical performance guarantees.

The problem of outlying sequence detection when all the typical distributions are identical was also studied previously as a universal outlier hypothesis testing problem for discrete samples in \cite{li2014universal} and for continuous samples in \cite{bu2016universal,zou2017anamoly}. In \cite{li2014universal}, the exponential consistency of the generalized likelihood (GL) test under various universal settings was established, where the GL test is based on computing the generalized likelihood function for each hypothesis by taking a maximum likelihood approach with respect to the unknown distributions. However, GL test cannot handle the scenario where the typical distributions are possibly different from each other, and so are the outlying distributions. Another major drawback of the GL test is its high time complexity, which is exponential in $M$ and $T$ (actually $M^T$ or $2^M$ depending on assumptions).

Our contributions in this paper are summarized as follows. We construct clustering-based tests that are exponentially consistent and have time complexity that is linear in $M$ for various scenarios. For the scenario where GL test is asymptotically optimal, we show that running more steps of the clustering-based test will not decrease the error exponent. We also show that the clustering-based tests are applicable to more general scenarios; for example, when both the typical and outlying distributions form clusters, the clustering-based test is exponentially consistent, but the GL test is not even applicable. We provide numerical results to demonstrate that the clustering-based tests can achieve an error exponent similar to that of the optimal test, but with time complexity that is linear in $M$. For all scenarios, our experiments indicate that running more steps of the clustering-based test results in a larger error exponent.

The rest of the paper is organized as follows. In Section \ref{sec:model}, we describe the problem model and three different test scenarios. In Section \ref{sec:GL}, we introduce the GL test studied in \cite{li2014universal}, which motivates the connection between universal outlying sequence detection and distribution clustering. In Section \ref{sec:cluster}, we reformulate the outlying sequence detection problem as a distribution clustering problem. In Section \ref{sec:new}, we propose linear-complexity tests based on the K-means clustering algorithm. In Section \ref{sec:num}, we provide numerical results. Finally in Section \ref{sec:con}, we conclude the paper.

\section{Problem Model}\label{sec:model}

%

Throughout the paper, all random variables are denoted by capital letters, and their realizations are denoted by the corresponding lower-case letters. All distributions are defined on the finite set $\mathcal{Y}$, and $\mathcal{P(Y)}$ denotes the set of all probability mass functions on $\mathcal{Y}$.

Consider an outlying sequence detection problem (see Fig.~\ref{fig:model}), where there are in total $M \ge 3$ data sequences denoted by $Y^{(i)}$ for $i=1,\dots, M$. Each data sequence $Y^{(i)}$ consists of $n$ i.i.d. samples $Y^{(i)}_1, \dots , Y^{(i)}_n$. The majority of the sequences are distributed according to typical distributions except for a subset $S$ of outlying sequences, where $S \subset \{1,\dots, M\}$ and $1\le |S| < \frac{M}{2}$.  Each typical sequence $j$ is distributed according to a typical distribution $\pi_j \in \mathcal{P(Y)},\ j\in S^C$. Each outlying sequence $i$ is distributed according to an outlying distribution $ \mu_i \in \mathcal{P(Y)},\ i \in S$. Nothing is known about $\mu_i$ and $\pi_j$ except that $\mu_i \ne \pi_j,\ \forall i\in S,\ j\in S^C,\ S \subset \{1,\dots, M\}$, and all of them have full support over $\mathcal{Y}$.
Denote $\mathcal{S}$ as the set comprising all possible outlying subsets.

\begin{figure}[!h]
	\centering
	\includegraphics[width=5.6cm]{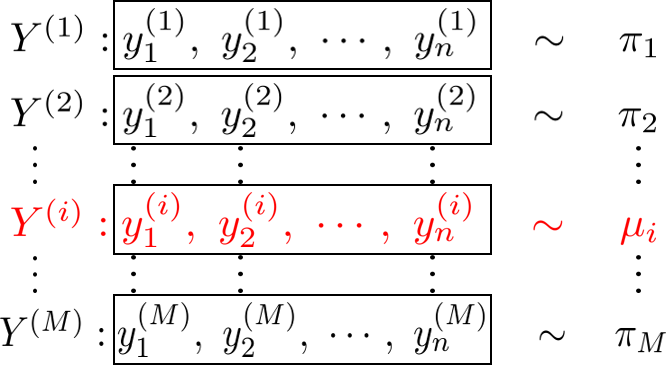}\\
	\caption{Outlying sequence detection with data sequences generated by typical
		distributions denoted by $\pi$ and outlying distributions denoted by $\mu$.}\label{fig:model}
\end{figure}

We use the notation $\bm{y}^{(i)} = (y^{(i)}_1, \dots , y^{(i)}_n)$, where $y^{(i)}_k \in \mathcal{Y}$ is the $k$-th observation of the $i$-th sequence.  Let $\gamma_i$ denote the empirical distribution of $\bm{y}^{(i)}$, and is defined as $$\gamma_i(y) \triangleq \frac{1}{n} \Big| \{k=1,\dots,n:y_k =y \} \Big|,$$ for each $y \in \mathcal{Y}$.

%

For the hypothesis corresponding to an outlying subset $S \in \mathcal{S}$, the joint distribution of all the observations is given by
\begin{equation}\label{eq:likelihood}
 p_S(y^{Mn})= L_S\Big(y^{Mn},\{\mu_i\}_{i\in S},\{\pi_j\}_{j \in S^C}\Big)=\prod_{k=1}^n\Big\{\prod_{i \in S } \mu_i(y_k^{(i)}) \prod_{j\in S^C } \pi_j(y_k^{(j)}) \Big\},
\end{equation}
where $L_S\big(y^{Mn},\{\mu_i\}_{i\in S},\{\pi_j\}_{j \in S^C}\big)$ denotes the likelihood.


Our goal is to build distribution-free tests to detect the outlying sequences. The test can be captured by a universal rule $\delta : \mathcal{Y}^{Mn} \to \mathcal{S}$, which must not depend on $\{\mu_i\}_{i\in S}$ and $\{\pi_j\}_{j \in S^C}$.

The performance of a universal test is gauged by the maximal probability of error, which is defined as
\begin{equation*}
  e\big(\delta\big) \triangleq \max_{S\in \mathcal{S}}
  \sum_{y^{Mn}:\ \delta(y^{Mn})\ne S} p_S(y^{Mn}),
\end{equation*}
and the corresponding error exponent is defined as
\begin{equation*}
  \alpha\big(\delta\big) \triangleq \lim_{n\to \infty}-\frac{1}{n}\log e\big(\delta\big).
\end{equation*}
\begin{definition}
A universal test $\delta$ is  said to be {universally exponentially consistent} if
\begin{equation*}
\alpha\big(\delta\big) >0.
\end{equation*}
\end{definition}

In this paper, we use $D(\pi\|\mu)$ and $B(\pi, \mu)$ to denote the {KL divergence} and {Bhattacharyya distance} between distributions $\pi$ and $\mu$:
\begin{flalign*}
D(\pi\|\mu) &\triangleq \sum_{y \in \mathcal{Y}} \pi(y) \log\left( \frac{\pi(y)}{\mu(y)}\right),\\
B(\pi,\mu)&\triangleq-\log\left(\sum_{y \in \mathcal{Y}} \pi(y)^{\frac{1}{2}}\mu(y)^{\frac{1}{2}}\right).
\end{flalign*}
\subsection{Three Scenarios}
In this paper, we focus on the following three scenarios.

We first study the scenario that all typical distributions are identical, i.e., $\pi_j=\pi,\ \forall j \in S^C$, and the number of outlying sequences is known at the outset, in Section \ref{subsec:3.1} and Section \ref{subsec:5.1}.

We next consider the scenario with identical typical distribution ($\pi_j=\pi,\ \forall j \in S^C$) and identical outlying distribution ($\mu_i=\mu,\ \forall i \in S$), and the number of outlying sequences is unknown, in Section \ref{subsec:3.2} and Section \ref{subsec:5.2}.

We note that without any further assumptions, when the number of the outlying sequences is unknown and outlying sequences can be distinctly distributed, there cannot exist a universally exponentially consistent test \cite{li2014universal}.

Thus, we then investigate the scenario where both the outlying distributions $\{\mu_i\}_{i\in S}$ and the typical distributions $\{\pi_j\}_{j \in S^C}$ form clusters in Section \ref{subsec:5.3}. Moreover, the typical distributions and the outlying distributions are distinct. More concretely,
\begin{align}\label{eq:clustercondition}
  \max_{i,j \in S} D(\mu_i\|\mu_j)  &< \min_{i \in S, j\in S^C} \{D(\mu_i\|\pi_j),  D(\pi_j\|\mu_i)\},\nn \\
  \max_{i,j \in S^C} D(\pi_i\|\pi_j)&< \min_{i \in S, j\in S^C} \{D(\mu_i\|\pi_j),  D(\pi_j\|\mu_i)\}.
\end{align}
This condition means that the divergence between two distributions within the same cluster is less than the divergence between two distributions in different clusters.

\section{Generalized Likelihood Test}\label{sec:GL}
In this section, we introduce the GL test for outlying sequence detection studied in \cite{li2014universal}, and summarize its consistency results (see \cite{li2014universal} for details). Here, it is assumed that the typical distributions are identical, i.e., $\pi_j=\pi,\ \forall j\in S^C$.

In the universal setting with $\pi$ and $\{\mu_i\}_{i\in S}$ unknown, conditioned on the outlying set being $S \in \mathcal{S}$, we compute the generalized likelihood of $y^{Mn}$ by replacing $\pi$ and $\{\mu_i\}_{i\in S}$ in \eqref{eq:likelihood} with their maximum likelihood estimates (MLEs) $\{\hat{\mu}_i\}_{i\in S}$, and $\hat{\pi}_S$, as
\begin{equation}
  \hat{p}_S^{\mathrm{univ}} = \hat{L}_S(y^{Mn},\{\hat{\mu}_i\}_{i\in S},\hat{\pi}_S).
\end{equation}
The GL test \cite{li2014universal} then selects the hypothesis under which the GL is maximized (ties are broken arbitrarily), i.e.,
\begin{equation}\label{eq:GLtest}
  \delta_{\mathrm{GL}}(y^{Mn}) = \argmax_{S \in \mathcal{S}} \hat{p}_S^{\mathrm{univ}}.
\end{equation}

In the following subsections, we consider different settings, where the suitable set $\mathcal{S}$ and the MLE of $\{\hat{\mu}_i\}_{i\in S}$, and $\hat{\pi}_S$ may differ.

\subsection{Known Number of Outlying Sequences}\label{subsec:3.1}
We first consider the scenario in which the number of outlying sequences, denoted by $T \ge 1$, is known at the outset, i.e., $\mathcal{S}=\{S: S\subset\{ 1,\dots,M\},|S|=T\}$. Moreover, the distributions of different outlying sequences $\mu_i, i \in S$, can be distinct from each other.

We compute the generalized likelihood of $y^{Mn}$ by replacing the $\mu_i,\ i \in S$ and $\pi$ in \eqref{eq:likelihood} with their MLEs: $$\hat{\mu}_i = \gamma_i, \ \text{ and }\ \hat{\pi}_S = \frac{\sum_{j \in S^C} \gamma_j}{M-T}.$$ Then, as in \cite{li2014universal}, the GL test in \eqref{eq:GLtest} is equivalent to
\begin{align}\label{eq:GLtest1}
  \delta_{\mathrm{GL}}(y^{Mn}) = &\argmin_{S \subset \mathcal{S}}  \sum_{j\in S^C} D\left(\gamma_j\bigg\|\frac{ \sum_{j\in S^C}\gamma_j}{M-T} \right ).
\end{align}

\begin{proposition}\label{prop:GLtest1} \cite[Theorem 10]{li2014universal}
Consider the scenario when the number of outlying sequences is known. If the typical distributions are identical, then the GL test in \eqref{eq:GLtest1} is universally exponentially consistent. As $M \to \infty$, the achievable error exponent converges as
%
\begin{equation*}
  \lim_{M \to \infty} \alpha\big(\delta_{\mathrm{GL}}\big) = \lim_{M \to \infty}\min_{i=1,\dots,M} 2B(\mu_i,\pi).
\end{equation*}
When all the outlying sequences are identically distributed, i.e., $\mu_i=\mu \ne \pi$, $i=1,\dots,M$,  the achievable error exponent of the GL test in \eqref{eq:GLtest1} converges to the optimal one achievable when both $\mu$ and $\pi$ are known, i.e., $2B(\mu,\pi)$.
\end{proposition}

Note that the number of hypotheses in the test \eqref{eq:GLtest1} is ${M \choose T}$. Thus, an exhaustive search over all possible hypotheses has time complexity that is polynomial in $M$ and exponential in $T$.


\subsection{Unknown Number of Identical Outlying Sequences}\label{subsec:3.2}
In this subsection, we consider the scenario where the number of outlying sequences is unknown, i.e., $\mathcal{S}=\{S: S\subset\{ 1,\dots,M\},1\le|S|<M/2\}$, and the hypotheses in $\mathcal{S}$ may have different numbers of outlying sequences. Moreover, it is assumed that the typical distributions are identical, and the outlying distributions are identical.

As shown in \cite{li2014universal}, by replacing the $\mu_i$, $i \in S$, and $\pi$ in \eqref{eq:likelihood} with their MLEs: $$\hat{\mu}_S = \hat{\mu}_i = \frac{\sum_{i\in S} \gamma_i}{|S|},\ \text{ and }\ \hat{\pi}_S = \frac{\sum_{j \in S^C} \gamma_j}{M-|S|},$$ the GL test in \eqref{eq:GLtest} is equivalent to
\begin{align}\label{eq:GLtest2}
  \delta_{\mathrm{GL}}(y^{Mn}) = \argmin_{S \subset \mathcal{S}}  \sum_{j\in S^C} D\left(\gamma_j\bigg\|\frac{ \sum_{j\in S^C}\gamma_j}{M-|S|} \right) + \sum_{i \in S} D\left(\gamma_i\bigg\|\frac{ \sum_{i\in S}\gamma_i}{|S|} \right ).
\end{align}
\begin{proposition} \cite[Theorem 11]{li2014universal}\label{prop:GLtest2}
Consider the scenario when the number of the outlying sequences is unknown, $1\le |S| <\frac{M}{2}$. If all the outlying
sequences are identically distributed, and all the typical sequences are identically distributed, the GL test in \eqref{eq:GLtest2} is universally exponentially consistent.
\end{proposition}

Note that the number of hypotheses in the GL test \eqref{eq:GLtest2} is $ \sum_{i=1}^{\lfloor M/2 \rfloor}{M \choose i},$ which is exponential in $M$. The time complexity of  \eqref{eq:GLtest2} is even larger than that of \eqref{eq:GLtest1}, without the knowledge of the number of outlying sequences.


Although the exponential consistency of the GL test under various universal settings was established in \cite{li2014universal}, the high time complexity, which is at least exponential in the number of outlying sequences $T$, limits its usage in applications.

In the following two sections, we reformulate the universal outlying sequence detection problem as a distribution clustering problem, and further propose clustering-based algorithms that are computationally efficient and exponentially consistent. In particular, we reduce the time complexity of our tests to $O(M)$, while still retaining a comparable error probability.



\section{Problem Reformulation as Distribution Clustering }\label{sec:cluster}


The GL test can be interpreted as combinatorial clustering over the probability simplex with the KL divergence as the distance measure.
More specifically, consider the problem of clustering distributions $p_1,p_2,\ldots$ into $K$ clusters, using a set of  cluster centers $c=\{c^{(1)}, \dots, c^{(K)}\}$, and a cluster assignment $C=\left\{C^{(1)}, \dots, C^{(K)}\right\}$. If we define the following cost function for distribution clustering
\begin{equation}
    TC \triangleq \sum_{k=1}^K \sum_{i \in C^{(k)}}  D(p_i\|c^{(k)}).
\end{equation}
As shown in \cite[Proposition 1]{banerjee2005clustering}, for a given cluster assignment $C=\left\{C^{(1)}, \dots, C^{(K)}\right\}$, the cost is minimized when
$$c^{(k)} = \frac{\sum_{i\in C^{(k)}} p_i}{|C^{(k)}|},$$
which is the average of the distributions within the $k$-th cluster. Thus, for a given cluster assignment  $C=\left\{C^{(1)}, \dots, C^{(K)}\right\}$, we have
\begin{equation}\label{eq:cost}
\min_{c^{(1)},\ldots,c^{(K)}}TC = \sum_{k=1}^K \sum_{i \in C^{(k)}}  D\left(p_i\bigg\|\frac{\sum_{i\in C^{(k)}} p_i}{|C^{(k)}|}\right).
\end{equation}

%


To connect the distribution clustering problem with the GL test, we first consider the scenario in Subsection \ref{subsec:3.2}, in which the typical distributions are identical and the outlying distributions are also identical. In view of \eqref{eq:cost}, the GL test in \eqref{eq:GLtest2} can be interpreted as a distribution clustering algorithm for the empirical distributions $\gamma_i$, $1\leq i\leq M$, with $K=2$ clusters. The first term in \eqref{eq:GLtest2} is the minimum cost in the typical cluster, and the second term is the minimum cost within the outlying cluster, for a given choice of $S$. The GL test then searches over all possible cluster assignments, and chooses the one with minimum cost.


We then consider the scenario in Subsection \ref{subsec:3.1}, where the typical distributions are identically distributed, but outliers are not (i.e., outlying distributions may not form a cluster).  We can utilize the knowledge of the number of outlying sequences, and it suffices to only cluster the empirical distributions of all the typical sequences, as shown in the GL test \eqref{eq:GLtest1}.


Thus, both the GL tests in \eqref{eq:GLtest1} and \eqref{eq:GLtest2} are equivalent to empirical distribution clustering on the probability simplex using KL divergence as the distance metric.

While the distribution clustering problem itself is known to be NP-hard \cite{chaudhuri2008finding}, there are many existing approximation algorithms with low time complexity, e.g., the K-means algorithm  \cite{lloyd1982least}.  Here, we introduce the K-means distribution clustering algorithm in Algorithm \ref{al:kmeans}, as proposed in \cite{banerjee2005clustering}.

\begin{algorithm}[!tb]
   \caption{K-means distribution clustering algorithm}\label{al:kmeans}
    \begin{algorithmic}
        \State \textbf{Input:} $M$ distributions $p_1$,\dots , $p_M$, defined on $\mathcal{Y}$, number of clusters $K$.
        \State \textbf{Output:} partition set $\{C^{(k)}\}_{k=1}^K$.
        \State \textbf{Initialization:} $\{c^{(k)}\}_{k=1}^K$ (Will be specified in Algorithms \ref{al:kmeans1} and \ref{al:kmeans2}.)
        \State \textbf{Method:}
        \While{not converge}
            \State \{Assignment Step\}
            \State Set $C^{(k)} \leftarrow \emptyset$, $1 \le k \le K$
            \For {$i = 1$ to $M$}
                \State $C^{(k^*)} \leftarrow C^{(k^*)} \cup \{p_i\}$
                \State where $k^*=\argmin_k D(p_i\|c^{(k)})$
            \EndFor
            \State \{Re-estimation Step\}
            \For {$k = 1$ to $K$}
            \vspace{0.15cm}
                \State $c^{(k)}\leftarrow\frac{\sum_{i\in C^{(k)}} p_i}{|C^{(k)}|}$
            \EndFor
        \EndWhile
    \State \textbf{Return} $\{C^{(k)}\}_{k=1}^K$
    \end{algorithmic}
\end{algorithm}

\begin{proposition}\label{prop:stop}\cite[Proposition 3]{banerjee2005clustering}
The cost function in \eqref{eq:cost} of Algorithm \ref{al:kmeans} is monotonically decreasing with steps. Moreover, Algorithm \ref{al:kmeans} terminates in a finite number of steps at a partition that is locally optimal, i.e., the total cost cannot be decreased by either (a) the assignment step, or (b) changing the means of any existing clusters.
\end{proposition}

\begin{remark}
The proof of Proposition \ref{prop:stop} in \cite{banerjee2005clustering} follows by the fact that the number of distinct cluster assignments is finite, and the fact that Algorithm \ref{al:kmeans} monotonically decreases the cost function in \eqref{eq:cost}. As shown in \cite{blomer2016theoretical}, for any Bregman divergence, the number of iterations of the K-means algorithm in the worst case can be upper bounded by $O(M^{K^2|\mathcal Y|})$, where $K$ is the number of clusters.

For our problem, the number of clusters is 2. Then, Algorithm \ref{al:kmeans} has polynomial time complexity in $M$ even in the worst case. In the following section, we will show that the exponential consistency can be established with a well-designed initialization in the first step, i.e., we do not need to wait for the algorithm to converge. As also will be shown in Section \ref{sec:num}, our tests usually converge in a few steps.
\end{remark}




\section{Clustering-Based Tests}\label{sec:new}
In this section, we propose linear-complexity tests based on the K-means clustering algorithm. We show in all three scenarios that the clustering-based tests using KL divergence as the distance metric are also exponentially consistent, while only taking linear time in $M$. For the scenario that the typical and outlying distributions form two clusters, we show that the clustering-based test is exponentially consistent, but the GL test is not even applicable.
\vspace{-0.1cm}
\subsection{Known Number of Outlying Sequences}\label{subsec:5.1}
We first consider the scenario where the number of outlying sequences $T$ is known and the typical distributions are identical.

\begin{algorithm}
   \caption{Clustering with known number of outlying sequences}\label{al:kmeans1}
    \begin{algorithmic}
        \State \textbf{Input:} $\gamma_1$,\dots , $\gamma_M$, number of the outlying sequences $T$.
        \State \textbf{Output:} A set of outlying sequences $S$.
        \State \textbf{Initialization:}
        \State Choose one distribution $\gamma^{(0)}$ from $\gamma_1$,\dots , $\gamma_M$ arbitrarily
            \For {$i = 1$ to $M$}
                \State Compute $D(\gamma_i\|\gamma^{(0)})$
            \EndFor
            \State $\hat{\pi} \leftarrow \gamma^*$
            \State where $D(\gamma^*\|\gamma^{(0)})$ is the $\lceil\frac{M}{2} \rceil$-th element among $D(\gamma_i\|\gamma^{(0)})$, $1\le i\le M$, arranged in an ascending order
        \State \textbf{Method:}
        \While{not converge}
            \State \{Assignment Step\}
            \State Set $S \leftarrow S^*$,
            \State where $S^*=\argmax_{S' \in \mathcal{S},|S'|=T} \sum_{i \in S' }D(\gamma_i\|\hat{\pi})$
            \State \{Re-estimation Step\}
            \State $\hat{\pi}\leftarrow\frac{\sum_{j\in S^C} \gamma_j}{M-T}$
        \EndWhile
    \State \textbf{Return} $S$
    \end{algorithmic}
\end{algorithm}
Note that Algorithm \ref{al:kmeans} cannot be directly applied here, because the outlying distributions may not form a cluster and Algorithm \ref{al:kmeans} does not employ the knowledge of $T$.

Motivated by the test in \eqref{eq:GLtest1}, we design Algorithm \ref{al:kmeans1}.
The novelty of this algorithm lies in the construction of the first cluster center for the typical distribution and the iterative approach based on K-means to update it.

By the initialization in Algorithm \ref{al:kmeans1}, $\gamma^*$ is generated from $\pi$ with high probability. The intuition behind this is that: if $\gamma^{(0)}$ is generated from the typical distribution $\pi$ as shown in Fig. \ref{fig:case1}, then only $|S|<\frac{M}{2}$ empirical distributions which are generated from $\mu_i$ are far from $\gamma^{(0)}$; if $\gamma^{(0)}$ is generated from some $\mu_i$ as shown in Fig. \ref{fig:case2}, then there are at least $M-|S|>\frac{M}{2}$ of $D(\gamma_i \| \gamma^{(0)})$ concentrating at $D(\pi\|\mu_i)$. Thus, the $\lceil\frac{M}{2} \rceil$-th element among all $D(\gamma_i\|\gamma^{(0)})$, $1\leq i\leq M$, arranged in an ascending order, is close to $D(\pi\|\mu_i)$, and $\gamma^*$ is generated from $\pi$ with high probability.

\begin{figure}[!h]
  \centering
  \subfigure[$\gamma^*$ is generated from the typical distribution]{
    \label{fig:case1} 
    \includegraphics[width=10cm]{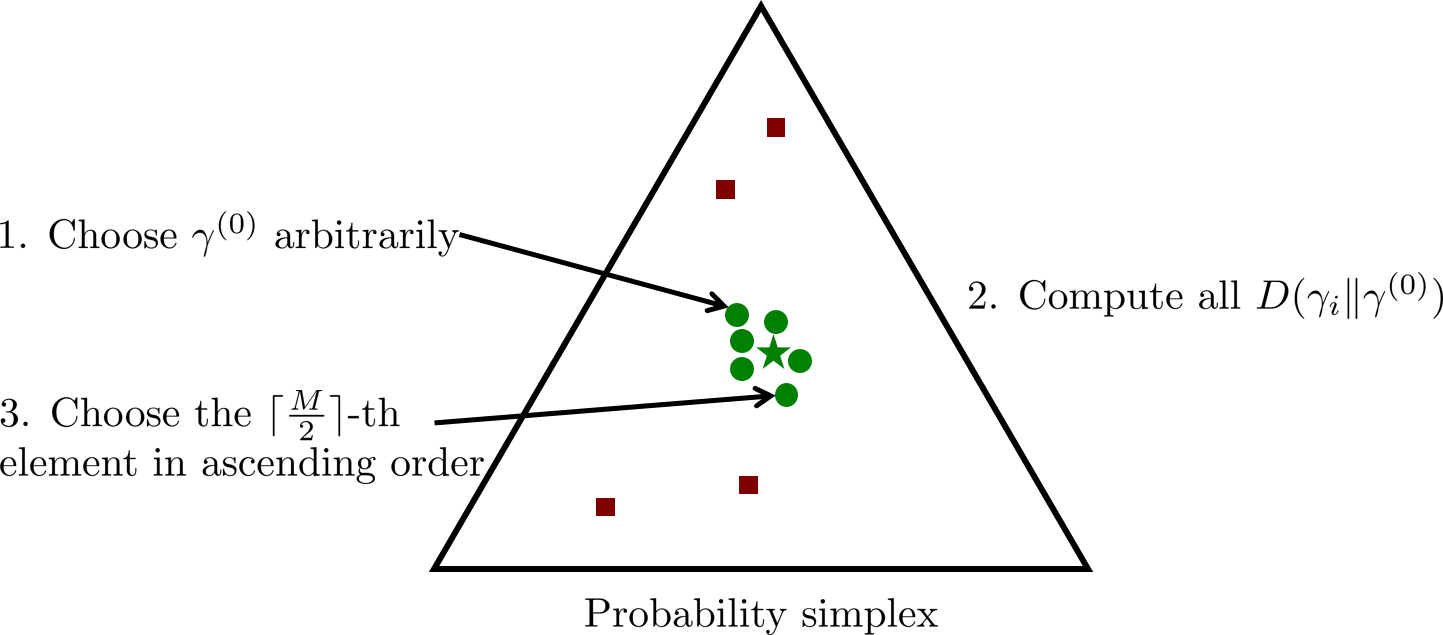}}
  \\
  \subfigure[$\gamma^*$ is generated from the outlying distribution]{
    \label{fig:case2} 
    \includegraphics[width=10cm]{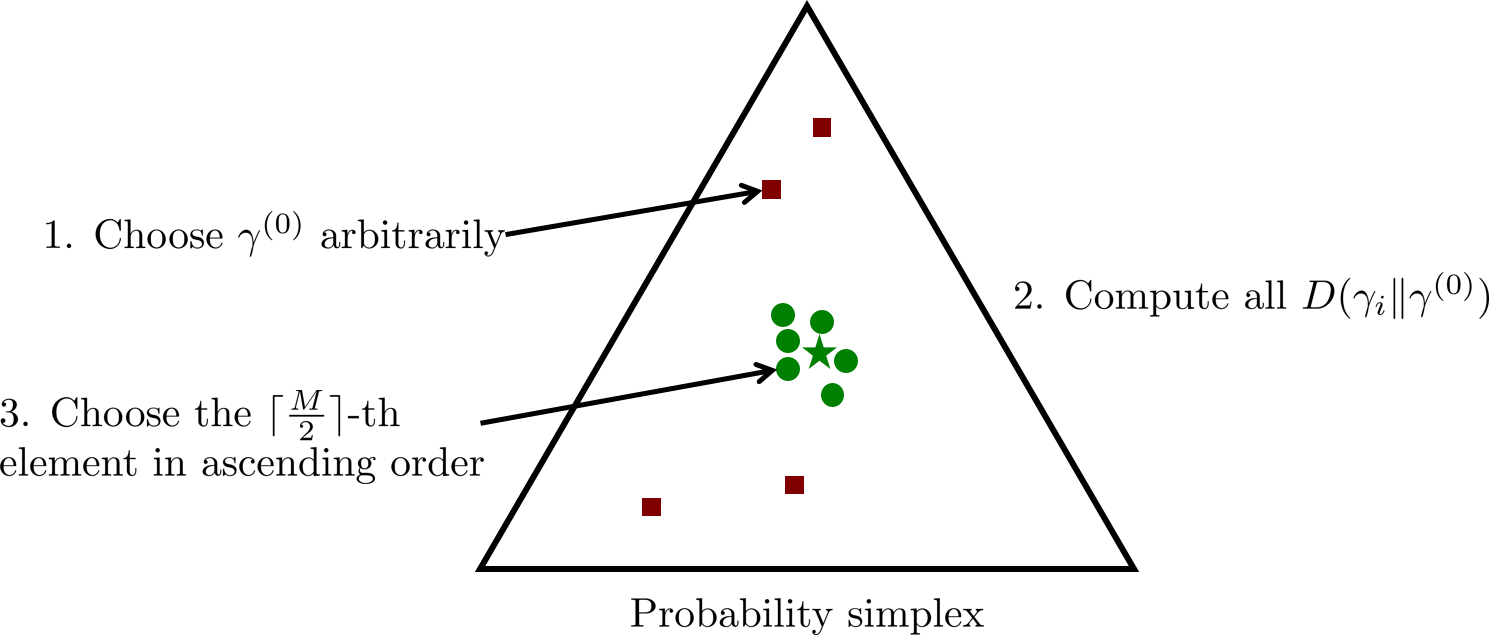}}
  \caption{Diagrams of Algorithm \ref{al:kmeans1}, where $M=10$, $|S|=4$. Green star denotes the identical typical distribution, and green circles (red squares) denote the  empirical distributions of typical (outlying) sequences, respectively. }
\end{figure}

Let $\delta_{2}$ denote the test described in Algorithm \ref{al:kmeans1}, and $\delta^{(\ell)}_{2}$ denote the test that runs $\ell$ number of K-means iterations in Algorithm \ref{al:kmeans1}.

In the following theorem, we show that the test $\delta^{(1)}_{2}$ (only one iteration step) is universally exponentially consistent.

\begin{theorem}\label{thm:kmeans1}
Consider the scenario when the number of outlying sequences $T$ is known. If the typical distributions are identical, the test $\delta^{(1)}_{2}$, which runs one K-means iteration in Algorithm \ref{al:kmeans1} is universally exponentially consistent. The achievable error exponent of $\delta^{(1)}_{2}$ can be upper bounded by
\begin{equation}
  \alpha\big(\delta^{(1)}_{2}\big) < \lim_{M \to \infty}\min_{i=1,\dots,M} 2B(\mu_i,\pi).
\end{equation}
Furthermore, the time complexity of the test $\delta^{(1)}_{2}$ is $O(M)$.
\end{theorem}

\begin{proof}[Proof sketch]
Errors made by $\delta^{(1)}_{2}$ in the initialization step can be decomposed into two scenarios. If $\gamma^{(0)}$ is generated from typical distribution $\pi$, an error occurs when $\gamma^*$ is actually generated from an outlying distribution. The probability of this event can be upper bounded by the probability of the following event
$$E_1 = \{D(\gamma_i\|\gamma_{j_1})<D(\gamma_{j_2}\|\gamma_{j_1}), \exists\ i \in S,\  \exists\  j_1,j_2 \in S^C \}.$$
If $\gamma^{(0)}$ is generated from an outlying distribution, the error probability can be upper bounded by the probability of the following event
$$E_2 =\{D(\gamma_{j_1}\|\gamma_{i_1})<D(\gamma_{i_2}\|\gamma_{i_1})<D(\gamma_{j_2}\|\gamma_{i_1}), \exists\ i_1,i_2 \in S,\  \exists\  j_1,j_2 \in S^C\}.$$
By Sanov's theorem \cite{cover2012elements}, we can prove that the probabilities of both $E_1$ and $E_2$ decay exponentially fast.

The error probability in the assignment step can be upper bounded by the probability of the same event $E_1$, which decays exponentially fast by Sanov's theorem.


Both the initialization and the assignment steps in Algorithm \ref{al:kmeans1} can be computed in linear time $O(M)$\cite{blum1973time}, independent of $T$.

The details of the proof can be found in Appendix \ref{appx:thm1}.
\end{proof}

A comparison between Proposition \ref{prop:GLtest1} and Theorem \ref{thm:kmeans1} shows that $\delta^{(1)}_{2}$ has a smaller error exponent than that of the GL test in \eqref{eq:GLtest1} as $M\to \infty$, but has a linear time complexity in $M$.

Although the exponential consistency can be established for the one step test $\delta^{(1)}_{2}$, it is of further interest to investigate whether the performance of Algorithm \ref{al:kmeans1} improves with more iterations.
In the following theorem, we show that the asymptotic performance of Algorithm \ref{al:kmeans1} does not decrease with more iterations. We will also see in our numerical results in Section  \ref{sec:num} that running more iterations of K-means always results in better performance.



\begin{theorem}\label{thm:moresteps}
For each $M \ge 3$, when the number of outlying sequences $T$ is known, the test $\delta^{(\ell)}_{2}$ is universally exponentially consistent. As $M\to \infty$, the achievable error exponent of $\delta^{(\ell)}_{2}$ in Algorithm \ref{al:kmeans1} can be lower bounded by
\begin{equation}
  \lim_{M\to \infty}\alpha\big(\delta^{(\ell)}_{2} \big) \ge \lim_{M\to \infty} \alpha\big(\delta^{(1)}_{2}\big).
\end{equation}
Furthermore, the time complexity of the test $\delta^{(\ell)}_{2}$ is $O(M\ell)$.
\end{theorem}


\begin{proof}
From Theorem \ref{thm:kmeans1} and Proposition \ref{prop:GLtest1}, we know that both the one-step test $\delta_{2}^{(1)}$ and the GL test $\delta_{\mathrm{GL}}$ are exponentially consistent. Then we have
\begin{align}
  \alpha\big(\delta^{(1)}_{2}\big) &= \lim_{n\to \infty} - \frac{1}{n} \log \mathbb{P}_S\Big(\delta^{(1)}_{2}(y^{Mn}) \ne S\Big)   ,\\
  \alpha \big(\delta_{\mathrm{GL}} \big) & = \lim_{n\to \infty} - \frac{1}{n} \log \mathbb{P}_S\Big(\delta_{\mathrm{GL}}(y^{Mn}) \ne S\Big).
\end{align}
Denote
\begin{equation}
  A \triangleq \big\{ y^{Mn}: \delta^{(1)}_{2}(y^{Mn}) \ne S \big\},\quad B \triangleq \big\{y^{Mn}: \delta_{\mathrm{GL}}(y^{Mn}) \ne S \big\}.
\end{equation}
Then the set $A \cup B$ contains those $y^{Mn}$ such that at least one of $\delta^{(1)}_{2}$ and $\delta_{\mathrm{GL}}$ makes an error.
Thus,
\begin{equation}
  \lim_{n\to \infty} - \frac{1}{n} \log \Big(\mathbb{P}_S\big(A \cup B\big)\Big) = \alpha\big(\delta^{(1)}_{2}\big)\wedge \alpha \big(\delta_{\mathrm{GL}} \big).
\end{equation}
This shows that the probability that at least one of $\delta^{(1)}_{2}$ and $\delta_{\mathrm{GL}}$ makes an error decays exponentially fast. Thus, the one-step test $\delta_{2}^{(1)}$ and the GL test $\delta_{\mathrm{GL}}$ output the same correct $S$ with high probability.

If $\delta_{2}^{(1)}$ and $\delta_{\mathrm{GL}}$ achieve the same outcome, which means both $\delta_{2}^{(1)}$ and $\delta_{\mathrm{GL}}$ have achieved the global minimum of the cost function \eqref{eq:cost}, then running more iterations in Algorithm \ref{al:kmeans1} will not change the outcome. Thus,
\begin{equation}
  e\big(\delta^{(\ell)}_{2} \big)\le \mathbb{P}_S\big(A \cup B\big).
\end{equation}
The error exponent of running $\ell$ iterations will be lower bounded by
\begin{equation}
  \alpha\big(\delta^{(\ell)}_{2} \big) = \lim_{n\to \infty} - \frac{1}{n} \log e\big(\delta^{(\ell)}_{2} \big) \ge \lim_{n\to \infty} - \frac{1}{n} \log \Big(\mathbb{P}_S\big(A \cup B\big)\Big) = \alpha\big(\delta^{(1)}_{2}\big)\wedge \alpha \big(\delta_{\mathrm{GL}} \big).
\end{equation}
As $M \to \infty$, we know that $\alpha\big(\delta^{(1)}_{2}\big)< \alpha \big(\delta_{\mathrm{GL}} \big)$ from Theorem \ref{thm:kmeans1}, then we can conclude that
\begin{equation}
  \lim_{M\to \infty}\alpha\big(\delta^{(\ell)}_{2} \big)
  \ge \lim_{M\to \infty}\alpha\big(\delta_{\mathrm{GL}}\big) \wedge  \alpha\big(\delta_{2}^{(1)}\big) \ge \lim_{M\to \infty}\alpha\big(\delta_{2}^{(1)}\big).
\end{equation}

As for the time complexity, since each iteration has time complexity $O(M)$, $\delta^{(\ell)}_{2}$ which runs $\ell$ iterations has time complexity $O(M\ell)$.
\end{proof}

\begin{remark}
There are other works applying different distance metrics in clustering algorithm, for example, using R\'enyi divergence \cite{poczos2012nonparametric} or maximum mean discrepancy
(MMD) \cite{zou2017anamoly}. However, the choice of KL divergence as the distance metric is crucial in our theoretical analysis, and the reason is two-fold:
\begin{enumerate}
  \item For a given cluster, the cost function with KL divergence is minimized when the center is the average of all distributions in this cluster \cite{banerjee2005clustering}. This property holds for all Bregman divergences and ensures that the clustering algorithm will terminate in finite steps. We have noticed that using MMD (not a Bregman divergence) will result in oscillating between two different clustering partitions.
  \item As shown in \cite{li2014universal}, when there are a known number of identically distributed outlying sequences, the GL test was shown to be asymptotically optimal as $M$ goes to infinity. It is based on this asymptotic optimality of using KL divergence that we prove that the error exponent of our clustering-based test is not decreasing as running more iterations.
\end{enumerate}
\end{remark}
\subsection{Unknown Number of Identical Outlying Sequences}\label{subsec:5.2}
In this section, we consider the scenario where the number of outlying sequences is unknown. Moreover, the typical distributions are identical, and the outlying distributions are identical.

\begin{algorithm}
   \caption{Clustering with unknown number of outlying sequences}\label{al:kmeans2}
    \begin{algorithmic}
        \State \textbf{Input:} $M$ empirical distributions $\gamma_1$,\dots , $\gamma_M$ defined on finite alphabet $\mathcal{Y}$.
        \State \textbf{Output:} A set of outlying sequences $S$.
        \State \textbf{Initialization:}
        \State Choose one distribution $\gamma^{(0)}$ arbitrarily,
            \State $c^{(1)} \leftarrow \argmax_{\gamma_i} D(\gamma_i\|\gamma^{(0)}) $
            \State $c^{(2)} \leftarrow \gamma^{(0)}$
        \State \textbf{Method:}
        Same as in Algorithm \ref{al:kmeans} with $K=2$
    \State \textbf{Return} $C^{(1)}$ and $C^{(2)}$
    \end{algorithmic}
\end{algorithm}

Since there is no prior information on the number of outlying sequences, we apply Algorithm \ref{al:kmeans} directly. Motivated by the test in \eqref{eq:GLtest2}, we design the following initialization in Algorithm \ref{al:kmeans2} to set the cluster centers in Algorithm \ref{al:kmeans}.

It can be seen that, with high probability, $c^{(1)}$ and $c^{(2)}$ chosen by the initialization step in Algorithm \ref{al:kmeans2} are generated by $\pi$ and $\mu$, respectively.


Let $\delta_{3}$ denote the test described in Algorithm \ref{al:kmeans2}, and $\delta^{(\ell)}_{3}$ denote the test that runs $\ell$ iterations in Algorithm \ref{al:kmeans1}. The following theorem shows that the clustering-based test $\delta^{(\ell)}_{3}$, is universally exponentially consistent, and has time complexity linear in $M$.

\begin{theorem}\label{thm:kmeans2}
Consider the scenario when the number of the outlying sequences is unknown, and $1\le |S| <\frac{M}{2}$. If all the outlying
sequences are identically distributed, and all the typical sequences are identically distributed, the test $\delta^{(\ell)}_{3}$, which runs $\ell$ steps of Algorithm \ref{al:kmeans2}, is exponentially consistent, and has time complexity $O(M\ell)$.
\end{theorem}

\begin{proof}[Proof sketch]
The exponential consistency of $\delta_{3}^{(\ell)}$ can be established using similar techniques to those in Theorem \ref{thm:kmeans1} and Theorem \ref{thm:moresteps}.
The major difference between the proof of Theorem \ref{thm:kmeans1} and Theorem \ref{thm:kmeans2} is that there are two cluster centers in the initialization step and assignment step in Algorithm \ref{al:kmeans2}.  The details can be found in Appendix \ref{appx:thm3}.
\end{proof}



\subsection{Typical and Outlying Distributions Forming Clusters}\label{subsec:5.3}

In this subsection, we consider the scenario that both the outlying distributions $\{\mu_i\}_{i\in S}$ and the typical distributions $\{\pi_j\}_{j \in S^C}$ are not identically distributed. Moreover, the typical distributions and the outlying distributions form clusters as defined in \eqref{eq:clustercondition}, which means that the divergence within the same cluster is always less than the divergence between different clusters.

The following theorem shows that under the condition \eqref{eq:clustercondition}, the one step test $\delta_{3}^{(1)}$ proposed in  Algorithm \ref{al:kmeans2} is universally exponentially consistent, and has time complexity that is linear in $M$.

\begin{theorem}\label{thm:clusters}
For each $M \ge 3$, when both the outlying distributions $\{\mu_i\}_{i\in S}$ and the typical distributions $\{\pi_j\}_{j \in S^C}$ form clusters, i.e. condition \eqref{eq:clustercondition} holds, the test $\delta_{3}^{(1)}$, which runs one step of Algorithm \ref{al:kmeans2}, is universally exponentially consistent, and has time complexity $O(M)$.
\end{theorem}
\begin{proof}[Proof sketch]
The exponential consistency of $\delta^{(1)}_{3}$ can be established using techniques similar to those in Theorem \ref{thm:kmeans2}. The details can be found in Appendix \ref{appx:clusters}.
\end{proof}

The GL approach of replacing the true distribution in \eqref{eq:likelihood} by their MLEs leads to identical likelihood estimates under each hypothesis. Thus, the GL approach is not applicable here. One could apply the test in \eqref{eq:GLtest2} to this problem, but the following example shows that the test in \eqref{eq:GLtest2} is not universally exponentially consistent, even if condition \eqref{eq:clustercondition} holds.

%
\begin{example}
As shown in \cite{li2014universal}, the error exponent of the GL test in \eqref{eq:GLtest2} is established by showing the following optimization problem has a positive value
\begin{align} \label{eq:exp_GL}
 \min_{q_1,q_2,\dots,q_M \in C_{(S,S')}} &\sum_{i\in S}D(q_i\|\mu_i)+\sum_{j\in S^C}D(q_j\|\pi_j), \nn \\
   \mbox{where }\quad \quad C_{(S,S')} = \Bigg\{(q_1,q_2,\dots,q_M): & \sum_{i\in S}D\Big(q_i\|\frac{\sum_{k \in S} q_k}{|S|}\Big) + \sum_{j\in S^C}D\Big(q_j\|\frac{\sum_{k \in S^C} q_k}{M-|S|}\Big) \nn \\
   &\ge\sum_{i\in S'}D\Big(q_i\|\frac{\sum_{k \in S'} q_k}{|S'|}\Big) + \sum_{j\notin S'}D\Big(q_j\|\frac{\sum_{k \notin S'} q_k}{M-|S'|}\Big) \Bigg\}.
\end{align}

We consider the scenario where $M =1000$, $S=\{1,2\}$, the typical and outlying distributions are specified as follows:
\begin{align}
\mu_1&=\left(\frac{1}{4},\frac{1}{2},\frac{1}{4}\right),\quad\quad \mu_2=\left(\frac{1}{5},\frac{7}{15},\frac{1}{3}\right),\nn \\
\pi_3&=\left(\frac{1}{3},\frac{1}{3},\frac{1}{3}\right),\quad\quad
\pi_4=\pi_5=\cdots=\pi_{1000}=\left(\frac{247}{500},\frac{32}{125},\frac{1}{4}\right).
\end{align}
It can be verified the clustering condition \eqref{eq:clustercondition} holds for this example. However, if we let
$q_1=\mu_1$, $q_2=\mu_2$, $q_3=\pi_3$, $q_4=\cdots=q_{1000}=\pi_4$, $S=\{1,2\}$ and $S'=\{1,2,3\}$,
then
\begin{equation}
\sum_{i\in S}D\Big(q_i\|\frac{\sum_{k \in S} q_k}{|S|}\Big) + \sum_{j\in S^C}D\Big(q_j\|\frac{\sum_{k \in S^C} q_k}{M-|S|}\Big)
   \ge\sum_{i\in S'}D\Big(q_i\|\frac{\sum_{k \in S'} q_k}{|S'|}\Big) + \sum_{j\notin S'}D\Big(q_j\|\frac{\sum_{k \notin S'} q_k}{M-|S'|}\Big)
\end{equation}
also holds, i.e., $(q_1,q_2,\dots,q_M)\in C_{S,S'}$, which means the error exponent in \eqref{eq:exp_GL} is equal to zero. Thus, the test in \eqref{eq:GLtest2} is not universally exponentially consistent for the scenario where both typical and outlying distributions form clusters.
\end{example}

\section{Numerical Results}\label{sec:num}
In this section, we compare the performance of the GL test $\delta_{\mathrm{GL}}$, the clustering-based tests $\delta_{2}$, $\delta_{3}$ (run until convergence) and the one step tests $\delta_{2}^{(1)}$, $\delta_{3}^{(1)}$.

For the scenario with identical typical distribution, we set $\pi$ to be the uniform distribution with alphabet size 10, and generate outlying distributions randomly.

We first simulate the scenario with distinct outlying distributions where $T$ is known. We choose $M=20$, $T=3$. In Figure \ref{fig:kmeans1}, we plot the logarithm of the error probability as a function of $n$ for $\delta_{\mathrm{GL}}$, $\delta_{2}$ and $\delta_{2}^{(1)}$ after 5000 Monte Carlo simulation steps. As we can see from Figure \ref{fig:kmeans1}, $\delta_{2}^{(1)}$ and $\delta_{2}$ are both exponentially consistent, and $\delta_{2}$ outperforms $\delta_{2}^{(1)}$ as shown in Theorem \ref{thm:moresteps}. A comparison of $\delta_{2}$ and $\delta_{\mathrm{GL}}$ shows that they have close performance, but $\delta_{2}$ is about 50 times faster than $\delta_{\mathrm{GL}}$.

We then simulate the scenario with unknown number of identical outlying distributions, setting $M=100$, $T=10$. Figure \ref{fig:kmeans2} shows that $\delta_{3}$ outperforms $\delta_{3}^{(1)}$. We note that running the clustering-based tests for 5000 Monte Carlo steps takes 5 minutes on a 3.6 GHz i7 CPU. However, the GL test is not feasible here, since the number of hypotheses one needs to search over is exponential in $M$.

For the scenario where both the typical and outlying distributions form clusters, we set the alphabet size to be 10. And we choose the uniform distribution as the center of the typical cluster. We generate the typical distributions by adding some Gaussian noise to the cluster center, and then normalizing them.  The cluster of the outlying distributions are generated in the same way, but with a randomly chosen cluster center. We set $M=100$, $T=10$. The dotted lines in Figure \ref{fig:cluster} correspond to the scenario where the typical distributions are identical and the outlying distributions are identical, and equal to the corresponding cluster center. The solid lines correspond to the scenario where both the typical and outlying distributions are generated by the approach mentioned above, which form clusters. Figure \ref{fig:cluster} shows that the tests $\delta_{3}^{(1)}$ and $\delta_{3}$ are exponentially consistent, and that $\delta_{3}$ outperforms $\delta_{3}^{(1)}$ for both scenarios.


We further study the number of iterations that $\delta_{3}$ takes to converge. Figure \ref{fig:steps} plots the average number of steps of test $\delta_{3}$ verses the number of samples, using the same setting as in Figure \ref{fig:kmeans2}. It is seen that the more the samples collected, the fewer the iterations needed. Moreover, when the number of samples $n$ goes to infinity, $\delta_{3}$ converges in just one step, which explains the exponential consistency of the one step test $\delta_{3}^{(1)}$.

Moreover, we study how the number of iterations that $\delta_{3}$ takes to converge changes as a function of the number of sequences $M$. We simulate the scenario where the number of outlying sequences $T$ changes with the number of sequences $M$ in Figure \ref{fig:change_M}. We use the same setting as in Figure \ref{fig:kmeans2}, but here we set $T=M/5$, where $M$ ranges from 40 to 200, and set $n=400$ to be fixed. Both Figure \ref{fig:steps} and Figure \ref{fig:change_M} show that in general our clustering-based test $\delta_{3}$ converges in very few steps.

\begin{figure}[!htp]
  \centering
  \includegraphics[width=8.2cm]{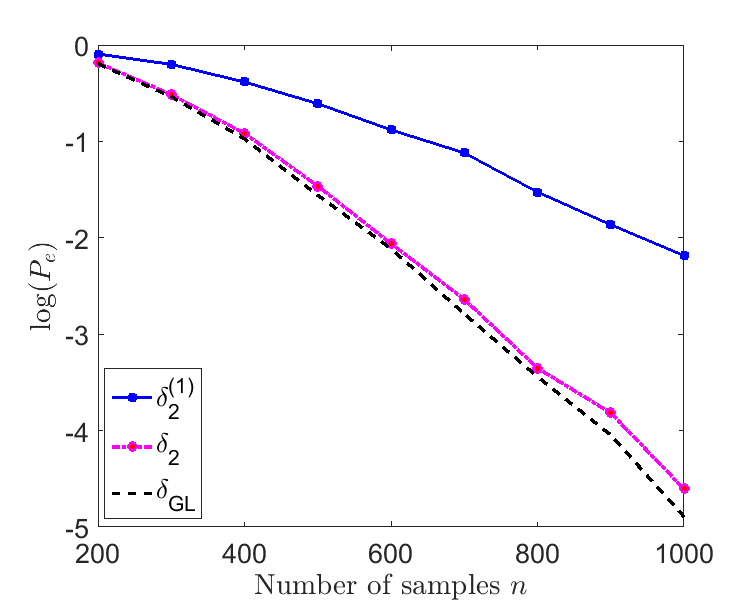}\\
  \caption{Comparison of tests  $\delta_{2}^{(1)}$, $\delta_{2}$, $\delta_{\mathrm{GL}}$ with known number of distinct outlying distributions.}\label{fig:kmeans1}
\end{figure}

\begin{figure}[!htp]
  \centering
  \includegraphics[width=8.2cm]{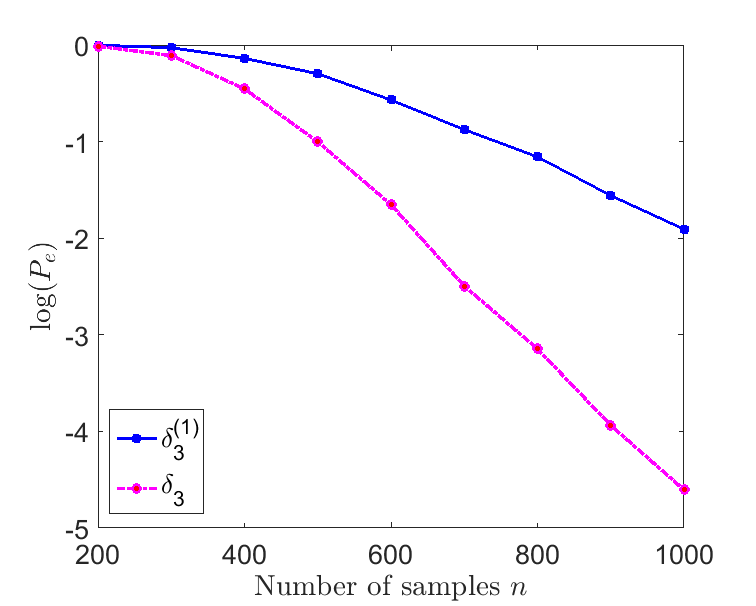}\\
  \caption{Comparison of tests $\delta_{3}$, $\delta_{3}^{(1)}$ with unknown number of identical outlying distributions.}\label{fig:kmeans2}
\end{figure}

\begin{figure}[!htp]
  \centering
  \includegraphics[width=8.2cm]{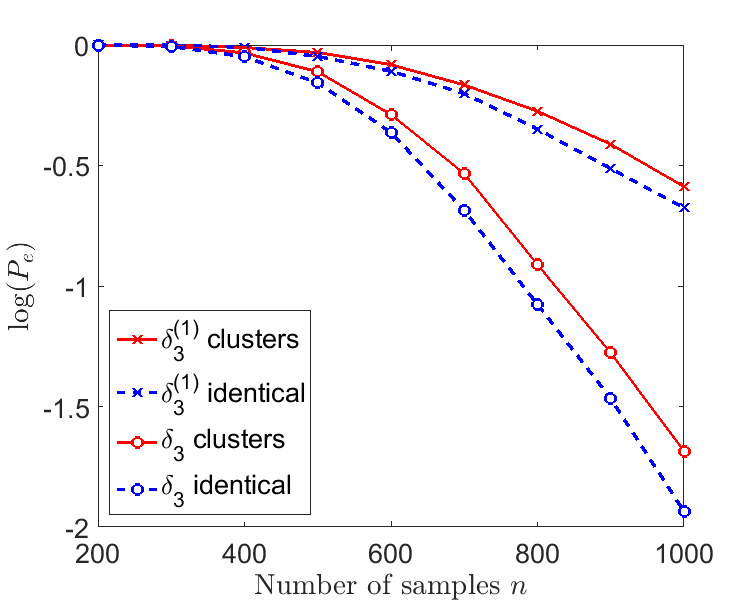}\\
  \caption{Comparison of tests $\delta_{3}$, $\delta_{3}^{(1)}$ when typical distributions and outlying distributions forming clusters respectively.}\label{fig:cluster}
\end{figure}

\begin{figure}[!htp]
  \centering
  \includegraphics[width=8.2cm]{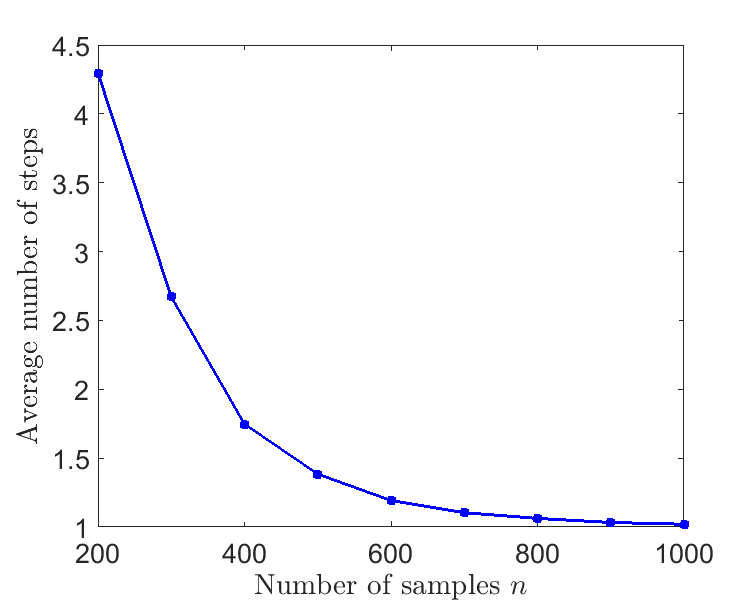}\\
  \caption{Average number of steps for convergence of test $\delta_{3}$ versus number of samples $n$.}\label{fig:steps}
\end{figure}

\begin{figure}[!htp]
  \centering
  \includegraphics[width=8.2cm]{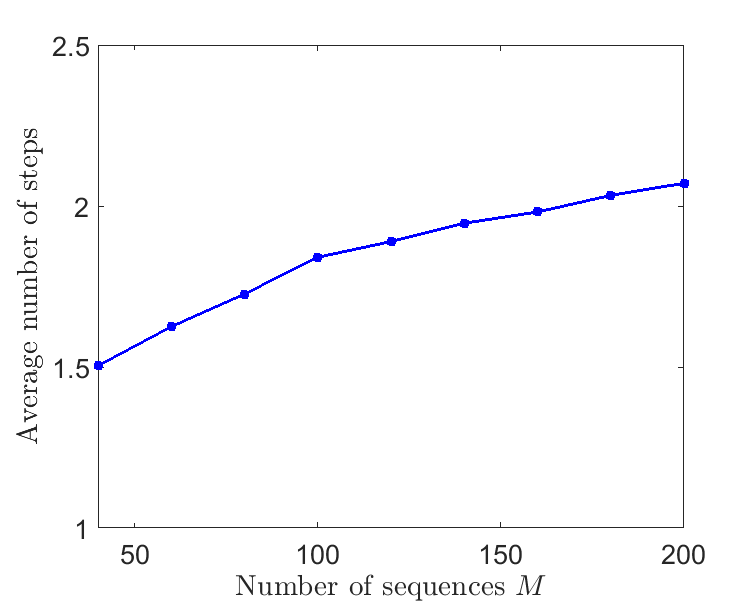}\\
  \vspace{-0.2cm}
  \caption{Average number of steps for convergence of test $\delta_{3}$ versus number of sequences $M$.}\label{fig:change_M}
  \vspace{-0.4cm}
\end{figure}

\section{Conclusion}\label{sec:con}

In this paper, we have investigated the universal outlying sequence detection problem. We have constructed clustering-based tests that are exponentially consistent and have time complexity linear in $M$ for various scenarios. For the scenario where GL test is asymptotically optimal, we have shown that running more steps of the clustering-based test does not decrease the error exponent. We have further  shown that the clustering-based tests are applicable to more general scenarios. For example, when both the typical and outlying distributions form clusters, the clustering-based test is exponentially consistent, but the GL test is not even applicable. We have provided numerical results to demonstrate that our clustering-based test can achieve a similar error exponent as the GL test.

Our study of this problem provides a new viewpoint to establish the exponential consistency of clustering algorithms, and we believe that this approach can be applied to universal outlying sequence detection with continuous distributions and also to other nonparametric problems.

\include{proof}

\end{document}

%% file: proof.tex
\newpage
\appendix
\noindent{\Large {\textbf{Appendix}}}
\section{Useful Lemmas}


\begin{lemma}\label{lemma:sanov}\cite[Lemma 1]{li2014universal}
Let ${Y}^{(1)}, \dots , {Y}^{(J)}$ be mutually independent random vectors with each ${Y}^{(j)}$, $j = 1, \dots , J$ , being $n$ i.i.d. samples of a random variable distributed according to $p_j \in \mathcal{P(Y)} $. Let $A_n$ be the set of all $J$ tuples $(\bm{y}^{(1)}, \dots , \bm{y}^{(J)}) \in \mathcal{Y}^{Jn}$ whose empirical distributions
$(\gamma_1, \dots ,\gamma_J )$ lie in a closed set $E \in \mathcal{P (Y)}^J$ . Then, it holds that
\begin{equation}
  \lim_{n \to \infty} -\frac{1}{n} \log \mathbb{P}\left\{\big( {Y}^{(1)}, \dots , {Y}^{(J)} \big)\in A_n \right\} = \min_{(q_1,\dots,q_J) \in E}\ \sum_{j=1}^J D(q_j \|p_j).
\end{equation}
\end{lemma}

\begin{lemma}\label{lemma:Bhatta}\cite[Lemma 2]{li2014universal}
For any two pmfs $p_1, p_2 \in \mathcal{P(Y)}$ with full supports, it holds that
\begin{equation}
  2B(p_1,p_2) = \min_{q \in \mathcal{P(Y)}} \Big(D(q\|p_1)+D(q\|p_2) \Big).
\end{equation}
In particular, the minimum on the right side is achieved by
\begin{equation}
  q^* = \frac{p_1^{1/2}(y)p_2^{1/2}(y)}{\sum_{y \in \mathcal{Y}}p_1^{1/2}(y)p_2^{1/2}(y)}, \ y\in \mathcal{Y}.
\end{equation}
\end{lemma}

\section{Proof of Theorem \ref{thm:kmeans1}} \label{appx:thm1}

Due to the structure of the test we know that errors may occur at the following two  steps:
\begin{enumerate}
  \item \textbf{Initialization Step:} The constructed cluster center for typical sequences $\hat{\pi}$ is actually generated from an outlying distribution.
  \item \textbf{Assignment Step:} Given that the cluster center $\hat{\pi}$ is actually generated from typical distribution $\pi$,  the empirical distribution of an outlying sequence is  closer to $\hat{\pi}$.
\end{enumerate}

We use $E$ to denote the event that errors occur in the initialization step. We use $\gamma^*$ as the cluster center $\hat{\pi}$, where $D(\gamma^*\|\gamma^{(0)})$ is the $\lceil \frac{M}{2} \rceil$-th element among all $D(\gamma_i\|\gamma^{(0)})$. It is difficult to write the explicit form of the event $E$. However, we can find upper bounds for the probability of $E$ for the following two scenarios.

If $\gamma^{(0)}$ is generated from the typical distribution $\pi$, an error occurs when $\hat{\pi}$ is actually generated from an outlying distribution, then $D(\gamma_i\|\gamma^{(0)})\le D(\gamma_{j}\|\gamma^{(0)})$ must hold for some $i \in S, j\in S^C$. Due to the arbitrariness of $\gamma^{(0)}$, the probability of this error event can be upper bounded by the probability of the following event:
\begin{equation}\label{eq:E1}
E_1 \triangleq \Big\{ D(\gamma_i\|\gamma_{j_1})\le D(\gamma_{j_2}\|\gamma_{j_1}),\ \exists\ i \in S,\  \exists\  j_1,j_2 \in S^C \Big\}.
\end{equation}

If $\gamma^{(0)}$ is generated from an outlying distribution, the error probability can be upper bounded by the probability of the following event:
\begin{equation}\label{eq:E2}
E_2 \triangleq \Big\{ D(\gamma_{j_1}\|\gamma_{i_1})<D(\gamma_{i_2}\|\gamma_{i_1})<D(\gamma_{j_2}\|\gamma_{i_1}),\ \exists\ i_1,i_2 \in S,\  \exists\  j_1,j_2 \in S^C \Big\}.
\end{equation}

Thus, $\mathbb{P}_S(E)\le \mathbb{P}_S(E_1)+\mathbb{P}_S(E_2)$.

We then use $F$ to denote the event that errors occur at the assignment step, then
\begin{equation}
  F \triangleq E^C\bigcap \Big\{ D(\gamma_i\|\hat{\pi})\le D(\gamma_{j}\|\hat{\pi}),\ \exists\ i\in S,\ \exists\ j\in S^C \Big\} .
\end{equation}
Note that $F\subset E_1$, then the probability of error event $F$ can be upper bounded by that of the event $E_1$.

The error probability of the test $\delta_{2}^{(1)}$ can be bounded by

\begin{equation}\label{eq:thm1_proof}
{\mathbb{P}_S( F) \le} e\big(\delta^{(1)}_{2} \big)  = \mathbb{P}_S\big(E\cup F\big) \le \mathbb{P}_S(E) +\mathbb{P}_S(F).
\end{equation}

The right hand side of \eqref{eq:thm1_proof} can be further bounded by
\begin{align}\label{eq:thm1_upper}
  & \mathbb{P}_S(E) +\mathbb{P}_S(F)\nn \\
  \le & \mathbb{P}_S(E_1)+\mathbb{P}_S(E_2)+ \mathbb{P}_S(E_1)\nn \\
  \le &   2 \mathbb{P}_S \Bigg( \bigcup_{\substack{j_1,j_2 \in S^C \\ i\in S}} \{D(\gamma_i\|\gamma_{j_1})\le D(\gamma_{j_2}\|\gamma_{j_1})\}\Bigg) \nn\\
   &\quad +  \mathbb{P}_S \Bigg(\bigcup_{\substack{j_1,j_2 \in S^C \\ i_1,i_2 \in S}}\{D(\gamma_{j_1}\|\gamma_{i_1})<D(\gamma_{i_2}\|\gamma_{i_1})<D(\gamma_{j_2}\|\gamma_{i_1})\}\Bigg)\nn \\
  \overset{(a)}{\le}& 2(M-|S|)^2|S|  \max_{i\in S} \mathbb{P}_S \Big(D(\gamma_i\|\gamma_{j_1})\le D(\gamma_{j_2}\|\gamma_{j_1})\Big) \nn\\
   &\quad + (M-|S|)^2|S|^2\max_{i_1,i_2\in S} \mathbb{P}_S \Big(D(\gamma_{j_1}\|\gamma_{i_1})<D(\gamma_{i_2}\|\gamma_{i_1})<D(\gamma_{j_2}\|\gamma_{i_1})\Big)\nn \\
  \le & (M-|S|)^2|S|^2 \bigg(2 \max_{i\in S} \mathbb{P}_S \Big(D(\gamma_i\|\gamma_{j_1})\le D(\gamma_{j_2}\|\gamma_{j_1})\Big) \nn\\
   &\quad + \max_{i_1,i_2\in S} \mathbb{P}_S \Big(D(\gamma_{j_1}\|\gamma_{i_1})<D(\gamma_{i_2}\|\gamma_{i_1})
  <D(\gamma_{j_2}\|\gamma_{i_1})\Big)\bigg),
\end{align}
where the union bound (a) holds for all $j_1,j_2 \in S^C$, since all typical sequences are generated from the same distribution $\pi$.

As for the left hand side of \eqref{eq:thm1_proof}, we have
\begin{align}\label{eq:thm1_lower}
 \mathbb{P}_S(F) \ge \max_{i\in S} \mathbb{P}_S\Big( D(\gamma_i\|\gamma_{j_1})\le D(\gamma_{j_2}\|\gamma_{j_1}) \Big),
\end{align}
for all $j_1, j_2 \in S^C$.

From Lemma \ref{lemma:sanov}, we know the exponent can be computed as
\begin{align} \label{eq:exp_1}
  \alpha_1  & \triangleq\lim_{n \to \infty} -\frac{1}{n} \log\ \max_{i \in S}\mathbb{P}_S \Big(D(\gamma_i\|\gamma_{j_1})\le D(\gamma_{j_2}\|\gamma_{j_1})\Big) \nn\\
& = \min_{\substack{q_1,q_2,q_3 \in C_ 1\\ i\in S}} D(q_1\|\mu_i)+D(q_2\|\pi) +D(q_3\|\pi),
\end{align}
where $C_1 \triangleq \big\{(q_1,q_2,q_3): D(q_1\|q_2)\le D(q_3\|q_2) \big\}$, and
\begin{align} \label{eq:exp_2}
  \alpha_2 & \triangleq\lim_{n \to \infty} -\frac{1}{n} \log\ \max_{i_1,i_2\in S} \mathbb{P}_S \Big(D(\gamma_{j_1}\|\gamma_{i_1})<D( \gamma_{i_2}\|\gamma_{i_1})
  <D(\gamma_{j_2}\|\gamma_{i_1})\Big) \nn \\
  &=\min_{\substack{q_1,q_2,q_3,q_4 \in C_ 2 \\  i_1,i_2 \in S}} D(q_1\|\pi)+D(q_2\|\pi) +D(q_3\|\mu_{i_1})+D(q_4\|\mu_{i_2}),
\end{align}
where $C_{2} \triangleq \big\{(q_1,q_2,q_3,q_4):D(q_1\|q_3)<D(q_4\|q_3) <D(q_2\|q_3)\big\}$.

It can be verified that the objective function in $\eqref{eq:exp_1}$ can
only be zero for the case $q_1=\mu_i$, $q_2 = q_3=\pi$,
which are not in the constraint set $C_1$. The objective function in $\eqref{eq:exp_2}$ can only be zero when  $q_1=q_2=\pi$, $q_3 =\mu_{i_1}, q_4=\mu_{i_2}$, which cannot meet the constraint in set $C_2$ either. Thus, we can conclude that $\alpha_1,\ \alpha_2 >0$.
From the fact that $\lim_{n \to \infty} \frac{\log M(M-|S|)}{n} =0$, and combining with \eqref{eq:thm1_upper} and \eqref{eq:thm1_lower}, we get that
\begin{equation}
   \alpha_1 \ge  \alpha\big(\delta^{(1)}_{2} \big) = \lim_{n \to \infty} -\frac{1}{n} \log e\big(\delta^{(1)}_{2} \big)  \ge \alpha_1\wedge \alpha_2.
\end{equation}
This result shows that the one step test $\delta_{2}^{(1)}$ is universally exponentially consistent.
We note that
\begin{equation}
  \alpha_1 = \min_{\substack{q_1,q_2,q_3 \in C_ 1\\ i\in S}} D(q_1\|\mu_i)+D(q_2\|\pi) +D(q_3\|\pi),
\end{equation}
where $C_1 = \big\{D(q_1\|q_2)\le D(q_3\|q_2) \big\}$. If we add the constrain that $q_1=q_3$, i.e., $C_1' = \big\{D(q_1\|q_2)\le D(q_3\|q_2),q_1=q_3=q \big\}$, then $C_1' \subset C_1$, and thus
\begin{align}
   \alpha_1 &  < \min_{\substack{q_1,q_2,q_3 \in C_ 1' \\ i\in S} } D(q_1\|\mu_i)+D(q_2\|\pi) +D(q_3\|\pi)\nn\\
   & = \min_{\substack{q,q_2 \in C_ 1'\\ i\in S}} D(q\|\mu_i)+D(q_2\|\pi) +D(q\|\pi)\nn\\
   & = \min_{\substack{q \in \mathcal{P(Y)}, i\in S}} D(q\|\mu_i) +D(q\|\pi).
\end{align}
Here, we can set $q_2 = \pi$, due to our relaxation to set $C_1'$. From Lemma \ref{lemma:Bhatta}, we know that the minimum is the Bhattacharyya distance between the distribution $\mu_i$ and $\pi$
\begin{equation}
   \alpha_1 < \min_{q \in \mathcal{P(Y)}, i\in S} D(q\|\mu_i) +D(q\|\pi) = \min_{ i\in S} B(\mu_i,\pi).
\end{equation}
Thus, as $M \to \infty$, we have
\begin{equation}
   \alpha\big(\delta^{(1)}_{2} \big) \le \alpha_1 < \lim_{M\to \infty} \min_{ i\in S} B(\mu_i,\pi).
\end{equation}

As for the time complexity, it is obvious that the initialization step in Algorithm \ref{al:kmeans1} can be executed within $O(M)$ time. The assignment step in Algorithm \ref{al:kmeans1}, which finds the largest $T$ elements from size $M$ array, can be solved in linear time $O(M)$ using the algorithm proposed in \cite{blum1973time}. Thus the overall time complexity is linear in $M$ and independent of $T$.

A comparison of Proposition \ref{prop:GLtest1} and Theorem \ref{thm:kmeans1} shows that $\delta^{(1)}_{2}$ has a smaller error exponent than that of the GL test in \eqref{eq:GLtest1} as $M\to \infty$, but has time complexity that is linear in $M$.

\section{Proof of Theorem \ref{thm:kmeans2}}\label{appx:thm3}
The exponential consistency of $\delta_{3}^{(\ell)}$ can be established using techniques similar to those in Theorem \ref{thm:kmeans1} and Theorem \ref{thm:moresteps}.
The major difference between the proof of Theorem \ref{thm:kmeans1} and Theorem \ref{thm:kmeans2} is that there are two cluster centers in the initialization step and assignment step in Algorithm \ref{al:kmeans2}.

We first establish the exponential consistency of the one-step test $\delta_{3}^{(1)}$.
Due to the structure of the test we know that errors may occur at two different steps:
\begin{enumerate}
  \item \textbf{Initialization Step:} The constructed cluster center for typical sequences $\hat{\pi}$ and outlying sequences $\hat{\mu}$ are actually generated from the same distribution.
  \item \textbf{Assignment Step:} The empirical distribution of an outlying sequence is closer to the cluster center of the typical sequence $\hat{\pi}$, and vice versa.
\end{enumerate}


We use $E$ to denote the event that errors occur in the initialization step. The error event $E$ can be decomposed into two parts, since $\gamma^{(0)}$ is chosen arbitrarily and can be generated from $\pi$ or $\mu$:
\begin{equation}
    E \triangleq E_1 \cup E_2,
\end{equation}
where
\begin{align}
  E_1 &\triangleq\Big\{\max _{j \in S^C} D(\gamma_j\| \gamma^{(0)}) > \max _{i \in S} D(\gamma_i\|\gamma^{(0)}),\ \gamma^{(0)}\ \mbox{generated from}\ \pi \Big\}, \\
  E_2 &\triangleq \Big\{\max _{i \in S} D(\gamma_i\| \gamma^{(0)}) > \max _{j \in S^C} D(\gamma_j\| \gamma^{(0)}),\ \gamma^{(0)}\ \mbox{generated from}\ \mu\Big\}.
\end{align}

Denote
\begin{align}
  A_i \triangleq \Big\{\max _{j_1 \in S^C} D(\gamma_{j_1}\|\gamma_{j_2} ) >  D(\gamma_i\| \gamma_{j_2}),\ \exists\ j_2 \in S^C \Big\},\ \forall\ i \in S, \\
  B_j \triangleq \Big\{\max _{i_i \in S} D(\gamma_{i_1}\| \gamma_{i_2}) >  D(\gamma_j\| \gamma_{i_2}),\ \exists\ i_2 \in S\Big\},\ \forall\ j \in S^C.
\end{align}

Since $\gamma^{(0)}$ is chosen arbitrarily, we have
\begin{align}
    E = \Big(\bigcap_{i \in S} A_i\Big) \bigcup \Big(\bigcap_{j \in S^C} B_j \Big).
\end{align}

We use $F$ to denote the event that errors occur at the assignment step, given that the clustering center $c^1$ and $c^2$ chosen by Algorithm \ref{al:kmeans2} are coming from different distributions. We further denote the cluster center which is actually generated from the typical (outlying) distribution by $\hat \pi$ ($\hat \mu$). Then $F$ can be written as
\begin{equation}
  F\triangleq F_1 \cup F_2,
\end{equation}
where
\begin{align}
  F_1 & \triangleq  E^C \bigcap\Big\{ D(\gamma_j\| \hat \pi)> D(\gamma_j\| \hat \mu), \exists\ j \in S^C \Big\},\\
  F_2 & \triangleq E^C \bigcap\Big\{D(\gamma_i\| \hat \mu) >  D(\gamma_i\| \hat \pi), \exists\ i\in S\Big\}.
\end{align}

Thus, we can upper bound the error probability of the one-step test $\delta^{(1)}_{3}$ by
\begin{equation}
  e\big(\delta^{(1)}_{3} \big)= \mathbb{P}_S\big(E \cup F \big) \le \mathbb{P}_S(E)+ \mathbb{P}_S(F).
\end{equation}

The first term on the right hand side can be bounded as,
\begin{align}
  \mathbb{P}_S(E) & \le \mathbb{P}_S\Big(\bigcap_{i \in S} A_i\Big)+ \mathbb{P}_S\Big(\bigcap_{j \in S^C} B_j\Big)\nn \\
    &\le \mathbb{P}_S( A_i) + \mathbb{P}_S( B_j)\nn \\
   & \overset{(a)}{\le} (M-|S|) \mathbb{P}_S \Big(\max _{j_1 \in S^C} D(\gamma_{j_1}\|\gamma_{j_2} ) >  D(\gamma_i\| \gamma_{j_2})\Big) \nn \\
   & \quad + |S| \mathbb{P}_S \Big(\max _{i_i \in S} D(\gamma_{i_1}\| \gamma_{i_2}) >  D(\gamma_j\| \gamma_{i_2})\Big)\nn \\
   & \overset{(b)}{\le} (M-|S|)^2 \mathbb{P}_S \Big(D(\gamma_{j_1} \| \gamma_{j_2}) > D(\gamma_i \| \gamma_{j_2} )\Big)\nn \\
   & \quad + |S|^2 \mathbb{P}_S \Big(D(\gamma_{i_1} \| \gamma_{i_2}) > D(\gamma_j \| \gamma_{i_2})\Big),
\end{align}
where the union bound (a) and (b) holds for all $j,j_1,j_2 \in S^C$ and $i,i_1,i_2 \in S$, since all typical distributions are identical and all outlying distributions are identical.

From Lemma \ref{lemma:sanov}, we obtain
\begin{align} \label{eq:exp_3}
  \alpha_3 &\triangleq \lim_{n \to \infty} -\frac{1}{n} \log \mathbb{P}_S \Big(D(\gamma_{j_1} \| \gamma_{j_2}) > D(\gamma_i \| \gamma_{j_2} )\Big)\nn \\ & =\min_{q_1,q_2,q_3 \in C_ 3} D(q_1\|\pi)+D(q_2\|\pi) +D(q_3\|\mu),
\end{align}
where $C_3 \triangleq \big\{(q_1,q_2,q_3): D(q_1\|q_2)>D(q_3\|q_2) \big\}$, and
\begin{align} \label{eq:exp_4}
  \alpha_4 & \triangleq \lim_{n \to \infty} -\frac{1}{n} \log \mathbb{P}_S \Big(D(\gamma_{i_1} \| \gamma_{i_2}) > D(\gamma_j \| \gamma_{i_2})\Big) \nn \\
  & =\min_{q_1,q_2,q_3 \in C_ 4} D(q_1\|\mu)+D(q_2\|\mu) +D(q_3\|\pi),
\end{align}
where $C_{4} \triangleq \big\{(q_1,q_2,q_3): D(q_1\|q_2)>D(q_3\|q_2) \big\}$.

We then upper bound $\mathbb{P}_S(F)$ by using the Union Bound \cite{durrett2010probability} as follows:
\begin{align}
  \mathbb{P}_S(F) & \le \mathbb{P}_S(F_1)+\mathbb{P}_S(F_2)\nn \\
   & \le \mathbb{P}_S\bigg( \bigcup_{j \in S^C }  \big\{D(\gamma_j\| \hat \pi)> D(\gamma_j\| \hat \mu) \big\} \bigg) \nn \\
   &\quad + \mathbb{P}_S\bigg( \bigcup_{i \in S } \big\{D(\gamma_i\| \hat \mu) >  D(\gamma_i\| \hat \pi) \big\}\bigg)\nn\\
   & \le |S|(M-|S|)^2 \mathbb{P}_S\Big( D(\gamma_{j_1}\| \gamma_{j_2})> D(\gamma_{j_1}\| \gamma_i)  \Big) \nn \\
   & \quad + |S|^2(M-|S|)\mathbb{P}_S\Big( D(\gamma_{i_1} \| \gamma_{i_2}) > D(\gamma_{i_1} \| \gamma_j)  \Big),
\end{align}
where $j,j_1,j_2 \in S^C$ and $i,i_1,i_2 \in S$. From Lemma \ref{lemma:sanov}, we obtain
\begin{align} \label{eq:exp_5}
  \alpha_5 & \triangleq \lim_{n \to \infty} -\frac{1}{n} \log \mathbb{P}_S\Big( D(\gamma_{j_1}\| \gamma_{j_2})> D(\gamma_{j_1}\| \gamma_i)  \Big) \nn \\
  & =\min_{q_1,q_2,q_3 \in C_ 5} D(q_1\|\pi)+D(q_2\|\pi) +D(q_3\|\mu),
\end{align}
where $C_5 \triangleq \big\{(q_1,q_2,q_3): D(q_1\|q_2)>D(q_1\|q_3) \big\}$, and
\begin{align} \label{eq:exp_6}
  \alpha_6 & \triangleq \lim_{n \to \infty} -\frac{1}{n} \log \mathbb{P}_S\Big( D(\gamma_{i_1} \| \gamma_{i_2}) > D(\gamma_{i_1} \| \gamma_j)  \Big)\nn \\
  & =\min_{q_1,q_2,q_3 \in C_ 6} D(q_1\|\mu)+D(q_2\|\mu) +D(q_3\|\pi),
\end{align}
where $C_{6} \triangleq \big\{(q_1,q_2,q_3): D(q_1\|q_2)>D(q_1\|q_3) \big\}$.

Due to the fact that the objective functions in $\eqref{eq:exp_3}$ and $\eqref{eq:exp_5}$ can
only be zero for the case $q_1 = q_2=\pi$, $q_3 =\mu$,
which is not in the constraint sets $C_3$ and $C_5$, respectively. The objective functions in $\eqref{eq:exp_4}$ and $\eqref{eq:exp_6}$ can only be zero when  $q_1=q_2=\mu$, $q_3=\pi$, which cannot meet the constraints in sets $C_4$ and $C_6$ either. Thus, we conclude that $\alpha_3,\ \alpha_4,\ \alpha_5,\  \alpha_6 >0$.

From the fact that $\lim_{n \to \infty} \frac{\log M(M-|S|)}{n} =0$, it then follows that
\begin{equation}
   \alpha\big(\delta_{3}^{(1)} \big) = \lim_{n \to \infty} -\frac{1}{n} \log e\big(\delta^{(1)}_{3} \big)  \ge \min \big\{\alpha_3,\ \alpha_4,\ \alpha_5,\  \alpha_6
   \big \}.
\end{equation}

From the above argument and Proposition \ref{prop:GLtest2}, both the one-step test $\delta_{3}^{(1)}$ and the GL test $\delta_{\mathrm{GL}}$ are exponentially consistent. Thus, based on the same technique used in the proof of Theorem \ref{thm:moresteps}, we establish the exponential consistency of the test $\delta_{3}^{(\ell)}$ proposed in Algorithm \ref{al:kmeans2}, for any $\ell\geq 1$.

Finally, since each iteration has the time complexity $O(M)$, $\delta^{(\ell)}_{3}$ which runs $\ell$ iterations has time complexity $O(M\ell)$.

\section{Proof of Theorem \ref{thm:clusters}}\label{appx:clusters}
The exponential consistency of $\delta_{3}^{(1)}$ for the scenario where the typical and outlying distributions form clusters can be established using the same techniques as in Theorem \ref{thm:kmeans2}. The major difference between the proofs of Theorem \ref{thm:kmeans2} and Theorem \ref{thm:clusters} is that here both the typical distributions and the outlying distributions are distinct.



Using the same events defined in Appendix \ref{appx:thm3}, the error probability of the one-step test $\delta^{(1)}_{3}$ can be upper bounded by
\begin{equation}
  e\big(\delta^{(1)}_{3} \big)= \mathbb{P}_S\big(E \cup F \big) \le \mathbb{P}_S(E)+ \mathbb{P}_S(F).
\end{equation}
The first term on the right hand side can be bounded as,
\begin{align}
  \mathbb{P}_S(E) & \le \mathbb{P}_S\Big(\bigcap_{i \in S} A_i\Big)+ \mathbb{P}_S\Big(\bigcap_{j \in S^C} B_j\Big)\nn \\
    &\le \mathbb{P}_S( A_i) + \mathbb{P}_S( B_j)\nn \\
   & \le (M-|S|)^2 \max_{\substack{ j_1,j_2\in S^C \\ i \in S}}\ \mathbb{P}_S \Big(D(\gamma_{j_1} \| \gamma_{j_2}) > D(\gamma_i \| \gamma_{j_2} )\Big) \nn \\
   & \quad + |S|^2 \max_{\substack{ i_1,i_2 \in S \\ j\in S^C}}\ \mathbb{P}_S \Big(D(\gamma_{i_1} \| \gamma_{i_2}) > D(\gamma_j \| \gamma_{i_2})\Big).
\end{align}

From Lemma \ref{lemma:sanov}, we obtain
\begin{align} \label{eq:exp_7}
   \alpha_7&\triangleq \lim_{n \to \infty} -\frac{1}{n} \log \max_{\substack{j_1,j_2\in S^C \\ i \in S}} \mathbb{P}_S \Big(D(\gamma_{j_1} \| \gamma_{j_2}) > D(\gamma_i \| \gamma_{j_2} )\Big) \nn \\
  & =\min_{\substack{j_1,j_2\in S^C\\ i \in S}}\ \min_{q_1,q_2,q_3 \in C_7} D(q_1\|\pi_{j_1})+D(q_2\|\pi_{j_2}) +D(q_3\|\mu_i),
\end{align}
where $C_7 \triangleq \big\{(q_1,q_2,q_3): D(q_1\|q_2)>D(q_3\|q_2) \big\}$, and
\begin{align} \label{eq:exp_8}
   \alpha_8 & \triangleq \lim_{n \to \infty} -\frac{1}{n} \log \max_{\substack{i_1,i_2 \in S\\ j\in S^C}} \mathbb{P}_S \Big(D(\gamma_{i_1} \| \gamma_{i_2}) > D(\gamma_j \| \gamma_{i_2})\Big) \nn \\
  & =  \min_{\substack{i_1,i_2 \in S \\ j\in S^C}}\ \min_{q_1,q_2,q_3 \in C_8} D(q_1\|\mu_{i_1})+D(q_2\|\mu_{i_2}) +D(q_3\|\pi_j),
\end{align}
where $C_{8} \triangleq \big\{(q_1,q_2,q_3): D(q_1\|q_2)>D(q_3\|q_2) \big\}$.

We then upper bound $\mathbb{P}_S(F)$ by using the Union Bound \cite{durrett2010probability} as follows,
\begin{align}
  \mathbb{P}_S(F) & \le \mathbb{P}_S(F_1)+\mathbb{P}_S(F_2)\nn \\
   & \le \mathbb{P}_S\bigg( \bigcup_{j \in S^C }  \big\{D(\gamma_j\| \hat \pi) > D(\gamma_j\| \hat \mu) \big\} \bigg)
   + \mathbb{P}_S\bigg( \bigcup_{i \in S } \big\{D(\gamma_i\| \hat \mu)  > D(\gamma_i\| \hat \pi) \big\}\bigg)\nn\\
   & \le |S|(M-|S|)^2 \max_{\substack{j_1,j_2\in S^C \\ i \in S}} \mathbb{P}_S\Big( D(\gamma_{j_1}\| \gamma_{j_2}) > D(\gamma_{j_1}\| \gamma_i) \Big)\nn \\
   &\quad+ |S|^2(M-|S|)\max_{\substack{i_1,i_2 \in S \\ j\in S^C}} \mathbb{P}_S\Big( D(\gamma_{i_1} \| \gamma_{i_2})  >D(\gamma_{i_1} \| \gamma_j) \Big).
\end{align}

From Lemma \ref{lemma:sanov}, we obtain
\begin{align} \label{eq:exp_9}
   \alpha_9 & \triangleq \lim_{n \to \infty} -\frac{1}{n} \log \max_{\substack{j_1,j_2\in S^C \\i \in S}} \mathbb{P}_S\Big( D(\gamma_{j_1}\| \gamma_{j_2}) > D(\gamma_{j_1}\| \gamma_i) \Big) \nn \\
  &= \min_{\substack{j_1,j_2\in S^C \\ i \in S}}\ \min_{q_1,q_2,q_3 \in C_ 9} D(q_1\|\pi_{j_1})+D(q_2\|\pi_{j_2}) +D(q_3\|\mu_i),
\end{align}
where $C_9 \triangleq \big\{ (q_1,q_2,q_3): D(q_1\|q_2)>D(q_1\|q_3) \big\}$, and
\begin{align} \label{eq:exp_10}
   \alpha_{10} &\triangleq \lim_{n \to \infty} -\frac{1}{n} \log \max_{\substack{i_1,i_2 \in S\\ j\in S^C}} \mathbb{P}_S\Big( D(\gamma_{i_1} \| \gamma_{i_2}) >D(\gamma_{i_1} \| \gamma_j) \Big) \nn \\
 &=\min_{\substack{i_1,i_2 \in S\\ j\in S^C}}\ \min_{q_1,q_2,q_3 \in C_{10}} D(q_1\|\mu_{i_1})+D(q_2\|\mu_{i_2}) +D(q_3\|\pi_j),
\end{align}
where $C_{10} \triangleq \big\{(q_1,q_2,q_3): D(q_1\|q_2)>D(q_1\|q_3) \big\}$.

Note that the objective functions in $\eqref{eq:exp_7}$ and $\eqref{eq:exp_9}$ can
only be zero for the case $q_1=\pi_{j_1}, q_2=\pi_{j_2}$, $q_3 =\mu_i$,
which is not in the constraint sets $C_7$ and $C_9$, due to our clustering assumption \eqref{eq:clustercondition}. The objective functions in $\eqref{eq:exp_8}$ and $\eqref{eq:exp_10}$ can only be zero when  $q_1=\mu_{i_1},q_2=\mu_{i_2}$, $q_3=\pi_j$, which cannot meet the constraints in sets $C_8$ and $C_{10}$ either. Thus, we conclude that $\alpha_7,\ \alpha_8,\ \alpha_9,\  \alpha_{10} >0$.

From the fact that $\lim_{n \to \infty} \frac{\log M(M-|S|)}{n} =0$, it follows that
\begin{equation}
   \alpha\big(\delta_{3}^{(1)}\big) = \lim_{n \to \infty} -\frac{1}{n} \log e\big(\delta^{(1)}_{3}\big)  \ge \min \big\{\alpha_7,\ \alpha_8,\ \alpha_9,\  \alpha_{10}\big\}.
\end{equation}

\bibliographystyle{IEEEbib}
\bibliography{Polynomial_test}